\documentclass[12pt,reqno]{amsart}
\usepackage{amsmath,amsfonts,amsthm,amssymb}
\newcommand\version{May 6, 2026}
\usepackage{amsxtra}

\newtheorem{theorem}{Theorem}

\newtheorem{lemma}{Lemma}

\theoremstyle{definition}

\newtheorem{example}{Example}

\theoremstyle{remark}

\newcommand{\C}{\mathbb{C}}

\newcommand{\eps}{\epsilon}
\renewcommand{\epsilon}{\varepsilon}

\newcommand{\R}{\mathbb{R}}

\DeclareMathOperator{\dom}{dom}
\DeclareMathOperator{\ran}{ran}

\begin{document}

\title{Arbitrary harmonic functions as Bose--Einstein condensates}

\author{Michiel De Wilde}
\address{Michiel De Wilde, Institute of Science and Technology Austria, Am Campus 1, 3400 Klosterneuburg, Austria}
\email{Michiel.DeWilde@ist.ac.at}

\author{Robert Seiringer}
\address{Robert Seiringer, Institute of Science and Technology Austria, Am Campus 1, 3400 Klosterneuburg, Austria}
\email{Robert.Seiringer@ist.ac.at}

\date\version

\thanks{\copyright\, 2026 by the authors. This paper may be  
reproduced, in
its entirety, for non-commercial purposes.}

\begin{abstract} 
We show that a suitable choice of boundary conditions for the Laplacian allows for the appearance of an arbitrary number of condensates, described by arbitrary harmonic functions, in the thermodynamic limit of an ideal Bose gas.
\end{abstract}


\maketitle


\section{Introduction and Main Result}

We shall investigate the question of what kind of condensates one can obtain for an ideal Bose gas in the thermodynamic limit if one is allowed to freely choose the boundary conditions for the Laplacian on the boundary of the domains. The question is motivated by the observation in \cite[Example~5.3.2]{BR} that ideal Bose gases in infinite volume allow for many different KMS states, with an arbitrary positive sesquilinear form that is invariant under the time evolution (separately in the arguments)  appearing in the two-point function, corresponding to harmonic functions describing the condensate(s). This leaves open the question whether such states can be obtained via the usual procedure of taking the thermodynamic limit of Gibbs states of systems in finite volume. 
Related to that, it was observed in \cite[Section~VI]{BY} that one can construct non-translation-invariant KMS states in the thermodynamic limit in such a way that the condensate is described by a linear function.

The following theorem shows that it is possible to choose boundary conditions for the Laplacian such that the corresponding system of an ideal Bose gas at zero chemical potential has, in the thermodynamic limit, any number of Bose condensates described by arbitrary harmonic functions.

\begin{theorem}\label{thm:main}
Consider a sequence of open, bounded domains $\Omega_j \subset \R^d$, with smooth boundary $\partial \Omega_j$, converging to $\R^d$ as $j\to \infty$, in the sense that any bounded set is contained in $\Omega_j$ for $j$ large enough. Given a (possibly infinite) sequence $\Phi = \{\phi_k\}_{k}$ of harmonic functions $\phi_k : \R^d \to \C$, with $\sum_k |\phi_k(x)|^2 < \infty$ for all $x\in\R^d$, there exist positive self-adjoint  realizations $-\Delta^{\Omega_j}_\Phi$ of $-\Delta$ on $\Omega_j$ such that 
\begin{equation}\label{eq:main}
\lim_{j \to \infty} \left\langle f \left| \frac 1{e^{-\beta \Delta_\Phi^{\Omega_j}} - 1} \right. g\right\rangle_{L^2(\Omega_j)}  = \int_{\R^d} \frac{ \overline{\hat f(p)} \hat g(p)}{e^{\beta |p|^2} -1} {dp}+ \frac 1\beta  \sum_{k}  \int_{\R^d} \overline { f}  \phi_k  \int_{\R^d}  \overline{\phi_k} g
\end{equation}
for any $\beta>0$ and any $f,g \in L^2(\R^d)$ with compact support. In dimensions $d\in\{1,2\}$, we need the additional assumption $\hat f(0) = \hat g(0) = 0$ for \eqref{eq:main} to hold.
\end{theorem}

By assumption, the supports of $f$ and $g$ are contained in $\Omega_j$ for $j$ large enough, hence the left-hand side of \eqref{eq:main} is well-defined. For the first term on the right-hand side to be well-defined, we need the additional assumption $\hat f(0) = \hat g(0) = 0$ for $d\in\{1,2\}$ to have integrability at $p=0$, which then follows from smoothness of $\hat f$ and $\hat g$. 

Eq.~\eqref{eq:main} corresponds to the two-point function of an ideal Bose gas with kinetic energy described by the Laplacian $-\Delta_\Phi^{\Omega}$ (see, e.g., \cite[Sect.~5.2.5]{BR} or \cite[Sect.~2.5]{Thirring}). From it, all $n$-point functions can be obtained via Wick's rule. For simplicity, we have set the chemical potential equal to zero, but the same result holds with a small (negative) chemical potential that goes to zero sufficiently fast as $j\to \infty$. 

There are multiple (possibly infinitely many) Bose--Einstein condensates. Each condensate comes with a corresponding density $\beta^{-1} |\phi_k(x)|^2$, whose mean value is infinite unless $\phi_k$ is the constant function, since non-constant harmonic functions necessarily increase at infinity. We remark that since the $\phi_k$ are harmonic, the pointwise summability condition $\sum_k |\phi_k|^2 < \infty$ necessarily holds uniformly on bounded sets.

The Laplacians $\Delta_\Phi^\Omega$ constructed in the next section will be strictly positive, but have very small eigenvalues as $\Omega$ gets large. In fact, from the convergence in \eqref{eq:main} one can read off that one needs to have eigenvalues that decrease as $(\int_\Omega |\phi_k|^2)^{-1}$ as $\Omega\to \R^d$, in order for the product of the inverse of the eigenvalue times the corresponding spectral projection to have a finite limit. 

By choosing a suitable sequence of domains $\Omega_j$, it is possible even in the case of Dirichlet boundary conditions  to have multiple condensates in the sense that more than one eigenvalue of the one-particle reduced density matrix (whose matrix elements are given by the two-point function above) is macroscopic (see \cite{berg1,berg2} and references therein). In the thermodynamic limit, this effect is not seen in the two-point function, however, in contrast to the case considered here.

The remainder of this article contains the proof of Theorem~\ref{thm:main}. In the next section, we shall construct the Laplacian $\Delta_\Phi^\Omega$ and describe its essential properties. In Section~\ref{sec:3}, we shall show that $\Delta_\Phi^{\Omega_j}$  converges to the  Laplacian on $\R^d$ as $j\to \infty$ in a suitable sense. Finally, Section~\ref{sec:4} combines the results of Sections~\ref{sec:2} and~\ref{sec:3} and gives the proof of Theorem~\ref{thm:main}.

\section{Construction of the Laplacians}\label{sec:2}

For nice enough domains $\Omega \subset \R^d$, let $\chi_\Omega$ denote the characteristic function of $\Omega$. The Dirichlet Laplacian on $\Omega$ will be denoted by $\Delta_0^\Omega$. Given a (possibly infinite) sequence of harmonic functions $\Phi = \{\phi_k\}_{k}$ on $\R^d$, with $\sum_k |\phi_k(x)|^2 < \infty$ for all $x\in\R^d$, we shall define the Laplacian $\Delta^\Omega_\Phi$ as 
\begin{equation}\label{def:delta}
-\Delta_\Phi^\Omega := \left( \left( -\Delta_0^\Omega\right)^{-1} + \sum_{k} |\chi_\Omega \phi_k \rangle \langle \phi_k \chi_\Omega| \right)^{-1}
\end{equation}
Note that the inverse of the Dirichlet Laplacian is compact, and so is the second term in parentheses in \eqref{def:delta} due to the summability assumption of the $\phi_k$. 
As the inverse of a compact operator with trivial kernel, this yields a well-defined self-adjoint operator. Moreover, we recover the Dirichlet Laplacian $-\Delta_0^\Omega$ in case $\Phi$ consists only of the zero function. 

Let $P_\Phi^\Omega : L^2(\Omega) \to L^2(\Omega)$ denote the projection onto the closure of the linear span of $\{ \chi_\Omega\phi_k\}_{k}$. Since the $\phi_k$ are assumed to be harmonic on all of $\R^d$, functions in the range of $P_\Phi^\Omega$ are smooth and extend smoothly to the boundary of $\Omega$ (and beyond).  We can introduce a positive (possibly unbounded) self-adjoint operator 
$R_\Phi^\Omega$ on $P_\Phi^\Omega L^2(\Omega)$ such that
\begin{equation}\label{def:R}
\sum_{k} |\chi_\Omega \phi_k \rangle \langle \phi_k \chi_\Omega|  = P_\Phi^\Omega \left( R_\Phi^\Omega\right)^{-1} P_\Phi^\Omega,
\end{equation}
as the left-hand side is positive and compact, and has a trivial kernel in $\ran P^\Omega_\Phi$.
%
%

For functions $f\in H^1(\Omega)$ that have a continuous restriction to the boundary $\partial \Omega$, we shall denote by $\psi_f$ the unique harmonic function that agrees with $f$ on the boundary of $\Omega$. 
The key properties of $-\Delta_\Phi^\Omega$ are summarized in the following lemma.

\begin{lemma}
With the definitions in \eqref{def:delta} and \eqref{def:R},
\begin{enumerate}
\item[(i)] $-\Delta_\Phi^\Omega$ is a positive self-adjoint extension of $-\Delta$ on $C_0^\infty(\Omega)$
\item[(ii)] The domain of $-\Delta_\Phi^\Omega$ is given by 
\begin{align}\nonumber
\dom(\Delta_\Phi^\Omega) = \big\{  f + g \, : \, &f \in H^2(\Omega)\cap H^1_0(\Omega), \ g \in \dom(R_\Phi^\Omega) \subset \ran P_\Phi^\Omega , \\ & R_\Phi^\Omega g =  P_\Phi^\Omega (-\Delta f) \big\}
\label{def:dom}
\end{align}
on which $-\Delta_\Phi^\Omega$ acts as $-\Delta_\Phi^\Omega(f+g) = - \Delta f$. 
%
\item[(iii)] $-\Delta_\Phi^\Omega$ is the operator associated to the quadratic form
$$
f\mapsto  \int_{\Omega} |\nabla f |^2 - \int_{\Omega} |\nabla \psi_f |^2  + \langle \psi_f | R_\Phi^\Omega \psi_f\rangle_{L^2(\Omega)}
$$
with form domain $\{ h + g \, : \, h \in H_0^1(\Omega), \ g \in \ran P_\Phi^\Omega, \  \langle \psi_g | R_\Phi^\Omega \psi_g\rangle_{L^2(\Omega)} < \infty\}$. 
\end{enumerate}
\end{lemma}

Functions  $f\in \dom(\Delta_\Phi^\Omega)$ have a smooth restriction to $\partial \Omega$, and $\psi_f \in \ran P_\Phi^\Omega$. Moreover, an integration by parts shows that $P_\Phi^\Omega (-\Delta f)$ only depends on $f$ and its outward normal derivative, $\partial f$, on the boundary of $\Omega$. 
That is, as $P_\Phi^\Omega (-\Delta f)$ is in the linear span of $\{\chi_\Omega\phi_k\}$, it suffices to know 
\begin{equation}\label{eq:int_by_parts}
    \langle\phi_k|-\Delta f\rangle_{L^2(\Omega)} = \int_{\partial\Omega} \left( -\phi_k \partial f +  f \partial\phi_k \right) \, .
\end{equation}
Hence the last condition in \eqref{def:dom}
is really a boundary condition. In fact, if we define an operator $r_\Phi^\Omega$ on $L^2(\partial \Omega)$ via 
$$
\langle \psi_f | R_\Phi^\Omega \psi_f\rangle_{L^2(\Omega)} = \langle f\! \!\restriction_{\partial\Omega} | r_\Phi^\Omega \, f\!\! \restriction_{\partial\Omega} \rangle_{L^2(\partial\Omega)} 
$$
for $\psi_f \in \ran P_\Phi^\Omega$, 
the boundary condition reads
\begin{equation}\label{bc}
\partial f 
= H^\Omega f\!\! \restriction_{\partial\Omega} - r_\Phi^\Omega\, f\!\!\restriction_{\partial\Omega}
\end{equation}
where $H^\Omega: L^2(\partial \Omega) \to L^2(\partial \Omega)$ is the Dirichlet-to-Neumann operator, mapping $f\!\!\restriction_{\partial\Omega}$ to $\partial \psi_f$. 

\begin{proof}
(i) Let us first check that $\Delta_\Phi^\Omega$ is indeed a Laplacian, i.e., $\Delta_\Phi^\Omega f = \Delta f$ for $f \in C_0^\infty(\Omega)$. By the definitions~\eqref{def:delta} and~\eqref{def:R} we have 
$$
\left(-\Delta_\Phi^{\Omega}\right)^{-1} \left(-\Delta_0^\Omega\right) f = f + P_\Phi^\Omega \left( R_\Phi^\Omega\right)^{-1} P_\Phi^\Omega \left(-\Delta_0^\Omega\right) f = f
$$ 
since $P_\Phi^\Omega \Delta_0^\Omega f = 0$, as the integration by parts in \eqref{eq:int_by_parts} shows for  $f$ and $\partial f$ vanishing on $\partial \Omega$. Hence $(-\Delta_\Phi^\Omega) f = (-\Delta_0^\Omega) f= -\Delta f$. Positivity and self-adjointness of $-\Delta_\Phi^\Omega$ follow straightforwardly from the definition. 

(ii) Next we compute the domain of $\Delta_\Phi^\Omega$. For $f\in \dom(\Delta^\Omega_\Phi)$, let $g = (-\Delta_\Phi^\Omega) f$. Again by the definitions~\eqref{def:delta} and~\eqref{def:R}, 
$$
f = (-\Delta_\Phi^\Omega)^{-1} g = \underbrace{ (-\Delta_0^\Omega)^{-1} g}_{f - \psi_f} + \underbrace{P_\Phi^\Omega \left(R_\Phi^\Omega\right)^{-1} P_\Phi^\Omega g}_{\psi_f} 
$$
Here we used the fact that  $(-\Delta_0^\Omega)^{-1} g \in H^2(\Omega) \cap H_0^1(\Omega)$, the domain of the Dirichlet Laplacian, and thus vanishes at the boundary of $\Omega$, while the last term $P_\Phi^\Omega (R_\Phi^\Omega)^{-1} P_\Phi^\Omega g$ is harmonic and smooth up to the boundary of $\Omega$. 
Hence $\psi_f \in \ran P_\Phi^\Omega$, and  $R_\Phi^\Omega \psi_f = P_\Phi^\Omega g = P_\Phi^\Omega (-\Delta f)$. 

(iii)
To obtain the quadratic form associated with $-\Delta_\Phi^\Omega$, an integration by parts using the boundary condition \eqref{bc} yields 
$$
\langle f | (-\Delta_\Phi^\Omega) f \rangle = \int_{\Omega} |\nabla f |^2 - \int_{\partial \Omega} \overline{f}  \partial f =  \int_{\Omega} |\nabla f |^2 - \int_{\Omega} |\nabla \psi_f |^2  + \langle \psi_f | R_\Phi^\Omega \psi_f\rangle_{L^2(\Omega)}
$$
together with the constraint $\psi_f \in \ran P_\Phi^\Omega$. 
\end{proof}

If, for given $\Omega$, one allows for harmonic functions that are merely harmonic on $\Omega$ instead of all of $\R^d$, one can show that all non-negative self-adjoint extensions of $-\Delta$ on $C_0^\infty(\Omega)$ can be obtained in the way above, plus suitably taking limits \cite{AS,B1,B2,B3}. In fact, any non-negative self-adjoint extension is bounded above by $-\Delta_0^\Omega$, and bounded below by the Krein Laplacian $Q^\Omega (Q^\Omega (-\Delta_0^\Omega)^{-1} Q^\Omega)^{-1} Q^\Omega$ with $1-Q^\Omega$ the projection onto the span of all harmonic functions on $\Omega$. 

For completeness, let us also note the general resolvent formula 
$$
\frac{1}{z+ \Delta_\Phi^\Omega} = \frac 1{z+ \Delta_0^\Omega} - \frac {\Delta_0^\Omega}{z+ \Delta_0^\Omega}  P_\Phi^\Omega \frac { 1}{ R_\Phi^\Omega - z P_\Phi^\Omega \frac {\Delta_0^\Omega}{z+ \Delta_0^\Omega} P_\Phi^\Omega } P_\Phi^\Omega \frac {\Delta_0^\Omega}{z+ \Delta_0^\Omega}
$$
which follows from \eqref{def:delta} and \eqref{def:R} and reduces to it for $z=0$. We shall not need it in the following, however. Formulas of this type are often called Krein resolvent formulas, following Krein's work on the difference of resolvents of different self-adjoint extensions \cite{krein}.

\begin{example}
As a specific example consider the case of $\Phi = \{ \phi\}$, i.e.,  there is just one harmonic function, denoted by $\phi$. In this case, we consider a Laplacian with domain given by functions $f = h + c\phi$, with $c\in \C$ and $h \in H^2(\Omega) \cap H_0^1(\Omega)$ satisfying Dirichlet boundary conditions. We have $\psi_f = c \phi$, and the boundary condition \eqref{bc} becomes 
$$
 f\!\!\restriction_{\partial\Omega} = c  \phi\!\!\restriction_{\partial\Omega}
\ , \ 
\partial f = c \partial \phi - r_\phi^\Omega\, f\!\!\restriction_{\partial\Omega}
$$ 
with $r_\phi^\Omega = 1 / \int_{\partial \Omega} |\phi|^2$. 
\end{example}

\section{Strong resolvent convergence}\label{sec:3}

In this section, we shall show that as $\Omega_j\to \R^d$ in the sense defined in Theorem~\ref{thm:main}, $-\Delta_\Phi^{\Omega_j}$ converges  to the usual Laplacian $-\Delta$ on $\R^d$ as $j\to \infty$ in a suitable sense. 
In fact, by proceeding similarly as in \cite[Thm.~1.2]{RT}, where a corresponding theorem for the Dirichlet Laplacian is given, we shall prove the following more general statement.

\begin{theorem}\label{thm:srs}
Consider a sequence of open, bounded domains $\Omega_j \subset \R^d$, converging to $\R^d$ as $j\to \infty$, in the sense that any bounded set is contained in $\Omega_j$ for $j$ large enough. Let $-\Delta^{\Omega_j}$ be any sequence of non-negative self-adjoint extensions of $-\Delta$ on $C^\infty_0(\Omega_j)$. If $F$ is a bounded continuous function on $[0,\infty)$, then  the operator $F(-\Delta^{\Omega_j})$ on $L^2(\Omega_j)$ composed with the restriction to $\Omega_j$ converges strongly to $F(-\Delta)$ on $L^2(\R^d)$, i.e. $F(-\Delta^{\Omega_j} ) \chi_{\Omega_j}  u \to F(-\Delta) u$ strongly in $L^2(\R^d)$ for all $u\in L^2(\R^d)$.  
\end{theorem}

\begin{proof}
By using the Stone--Weierstrass theorem, we can argue as in the proof of \cite[Thm.~1.2]{RT} to conclude that it suffices to prove the statement for the function $F(x) = (1+x)^{-1}$ as it separates points on $[0,\infty)$ and vanishes nowhere. Therefore, it generates the algebra of bounded continuous functions vanishing at infinity. A further approximation argument then allows for the extension to all bounded continuous functions.

Take $u \in L^2(\R^d)$, and let $v_j = (-\Delta^{\Omega_j}+1)^{-1}\chi_{\Omega_j} u $. Since $-\Delta^{\Omega_j}$ is non-negative, $v_j \in L^2(\R^d)$ with $ \| v_j\|_{L^2(\R^d)  }=\| v_j\|_{L^2(\Omega_j)} \leq \|u\|_{L^2(\R^d)}$.
    By selecting a subsequence, we may hence assume that there is a weak $L^2$ limit $v_j\rightharpoonup w$. 
 
    For $\psi \in C^\infty_0(\R^d)$ we have  for sufficiently large $j$
    $$
    \langle\psi| -\Delta^{\Omega_j} v_j \rangle_{L^2(\Omega_j)} = \langle -\Delta^{\Omega_j} \psi|v_j\rangle_{L^2(\Omega_j)} = \langle -\Delta \psi|v_j\rangle_{L^2(\R^d)} \xrightarrow{j\to \infty} \langle -\Delta\psi|w \rangle_{L^2(\R^d)}   
    $$
    where we used that $\Delta^{\Omega_j}$ is a  self-adjoint extension of $\Delta$ on $C^\infty_0(\Omega_j)$. On the other hand,
    $$
    \langle\psi| -\Delta^{\Omega_j} v_j \rangle_{L^2(\Omega_j)} =\langle\psi| \chi^{\Omega_j} u \rangle_{L^2(\Omega_j)} - \langle\psi|  v_j \rangle_{L^2(\Omega_j)} \xrightarrow{j\to \infty} \langle \psi|  u\rangle_{L^2(\R^d)}  - \langle \psi|w \rangle_{L^2(\R^d)}  
    $$
  and hence we conclude that $-\Delta w + w = u$ in the sense of distributions. Since $w\in L^2(\R^d)$, this implies that $w=(-\Delta + 1)^{-1} u$. 
    
 We are left with showing strong convergence. For $\eta \in \C_0^\infty(\R^d)$, we can bound $\| v_j - w \|_{L^2(\R^d)}  \leq \| v_j - \eta w\|_{L^2(\R^d)} + \| (1-\eta) w \|_{L^2(\R^d)}$, and the last term can be made arbitrarily small by choosing $\eta$ to be equal to $1$ on a large enough ball. For $j$ large enough, the support of $\eta$ is contained in $\Omega_j$, and hence 
 $$
 \| v_j - \eta w\|_{L^2(\R^d)} = \| v_j - \eta w\|_{L^2(\Omega_j)} \leq  \| ( 1 - \Delta^{\Omega_j}) (v_j - \eta w)\|_{L^2(\Omega_j)}
 $$
 using again positivity of $-\Delta^{\Omega_j}$. By the definition of $v_j$ and using $(-\Delta + 1) w = u$,
 $$
 \left( 1 - \Delta^{\Omega_j}\right) (v_j - \eta w) = \left( \chi_{\Omega_j} -\eta \right) u  + 2 \nabla \eta \cdot \nabla w + \left( \Delta \eta\right) w
 $$
 and this converges strongly to $(1-\eta) u + 2 \nabla\eta\cdot \nabla w + (\Delta \eta) w$ as $j\to \infty$. Since $w\in H^1(\R^d)$, this term can again be made arbitrarily small by a suitable choice of $\eta$. This completes the proof.
\end{proof}

\section{Proof of Theorem~\ref{thm:main}}\label{sec:4}

In this section we shall give the proof of Theorem~\ref{thm:main}. By the polarization identity, we can assume $f=g$ in \eqref{eq:main} without loss of generality. 

Pick $j$ large enough such that $\Omega_j$ contains the support of $f$. With $$F(x) = 1/(e^x-1) - 1/x\,,$$ we can write 
$$
\left\langle f \left| \frac 1{e^{-\beta \Delta_\Phi^{\Omega_j}} - 1} \right. f\right\rangle_{L^2(\Omega_j)} = \left\langle f \left| F(-\beta \Delta_\Phi^{\Omega_j}) \right. f\right\rangle_{L^2(\Omega_j)} + \beta^{-1} \left\langle f \left|  \left(-\Delta_\Phi^{\Omega_j}\right)^{-1} \right. f\right\rangle_{L^2(\Omega_j)}
$$
Since $F$ is a bounded continuous function on $[0,\infty)$, we deduce from Theorem~\ref{thm:srs} that 
$$
\lim_{j\to \infty} \left\langle f \left| F(-\beta \Delta_\Phi^{\Omega_j}) \right. f\right\rangle_{L^2(\Omega_j)} = \left\langle f \left| F(-\beta \Delta) \right. f\right\rangle_{L^2(\R^d)} = \int_{\R^d} F(\beta |p|^2) |\hat f(p)|^2  {dp}
$$ 
From the definition of $\Delta_\Phi^{\Omega}$ in \eqref{def:delta}, we further have 
$$
\left\langle f \left|  \left(-\Delta_\Phi^{\Omega_j}\right)^{-1} \right. f\right\rangle_{L^2(\Omega_j)} = \left\langle f \left|  \left(-\Delta_0^{\Omega_j}\right)^{-1} \right. f\right\rangle_{L^2(\Omega_j)} + \sum_k \left| \int \overline{\phi_k} f \right|^2 
$$
The statement of the theorem thus follows if we can show that 
$$
\lim_{j\to \infty} \left\langle f \left|  \left(-\Delta_0^{\Omega_j}\right)^{-1} \right. f\right\rangle_{L^2(\Omega_j)} = \left\langle f \left|  \left(-\Delta\right)^{-1} \right. f\right\rangle_{L^2(\R^d)} = \int_{\R^d} \frac{ |\hat f(p)|^2}{|p|^2} dp 
$$
But this is a consequence of Theorem~\ref{thm:srs} and the monotone convergence theorem, using that $\langle f | (-\Delta_0^{\Omega} + \eps)^{-1} f\rangle_{L^2(\Omega)}$ is increasing as $\eps\to 0^+$ and  $\Omega \to \R^d$, i.e., for all $f$ supported inside $\Omega$, $\langle f | (-\Delta_0^{\Omega} + \eps)^{-1} f\rangle_{L^2(\Omega)} \leq \langle f | (-\Delta_0^{\Omega'} + \eps')^{-1} f\rangle_{L^2(\Omega')}$ if $\eps \geq \eps'>0$ and $\Omega \subset \Omega' \subset \R^d$. This completes the proof. 

\bigskip
{\it Acknowledgments.} We are grateful to Rupert Frank and Jakob Yngvason for  helpful discussions and suggestions.


\begin{thebibliography}{19}

\bibitem{AS} A. Alonso, B. Simon, {\it The Birman--Kre{\u\i}n--Vishik theory of self-adjoint extensions of semibounded operators}, J. Operator Theory {\bf 4}, 251--270 (1980)

\bibitem{berg1} M. van den Berg, J.T. Lewis, {\it On generalized condensation in the free boson gas}, Physica {\bf 110A}, 550--564 (1982)

\bibitem{berg2} M. van den Berg, J.T. Lewis, J.V. Pul\'e, {\it A general theory of Bose--Einstein condensation}, Helv. Phys. Acta {\bf 59}, 1271--1288 (1986)

\bibitem{B1} M.S. Birman, {\it On the self-adjoint extensions of positive definite operators (Russian)}, Math. Sb. {\bf 38}, 431--450 (1956)

\bibitem{BR} O. Bratteli, D.W. Robinson, {\it Operator Algebras and Quantum Statistical Mechanics 2. Equilibrium States. Models in Quantum Statistical Mechanics}, $2^{\rm nd}$ ed., Springer (1997) 

\bibitem{BY}  D. Buchholz, J. Yngvason, {\it Many-body physics and resolvent algebras}, 
J. Math. Phys. {\bf 66}, 051902 (2025)

\bibitem{krein} M.G. Krein, {\it Concerning the resolvents of an Hermitian operator with the deficiency-index $(m, m)$}, Comptes Rendue (Doklady) Acad. Sci. URSS (NS) {\bf 52}, 651--654 (1946)

\bibitem{B2} M. Krein, {\it The theory of self-adjoint extensions of semibounded Hermitian transformations and its applications. I}, Rec. Math. (Math. Sb.) {\bf 20} (62), 431--495 (1947)

\bibitem{RT} J. Rauch, M. Taylor, {\it Potential and Scattering Theory on Wildly Perturbed Domains}, J. Funct. Anal. {\bf 18}, 27--59 (1975)

\bibitem{Thirring} W. Thirring, {\it Quantum Mathematical Physics}, $2$nd ed., Springer (2002)

\bibitem{B3} M. Vishik, {\it On general boundary conditions for elliptic differential equations (Russian)}, Trudy Moskov. Mat. Obsc. {\bf 1}, 187--246 (1952)

\end{thebibliography}
\end{document}